\documentclass[conference]{IEEEtran}

\usepackage{amsmath, amsthm, amssymb}
\usepackage{latexsym}
\usepackage{graphicx}
\usepackage{verbatim}
\usepackage{mathrsfs}
\usepackage{algorithmic}
\usepackage{float}
\usepackage[lofdepth, lotdepth]{subfig}
\usepackage{array}
\usepackage{url}
\usepackage{cases}

\theoremstyle{plain}
\newtheorem{theorem}{Theorem}
\newtheorem{lemma}[theorem]{Lemma}

\newtheorem{corollary}[theorem]{Corollary}

\theoremstyle{definition}
\newtheorem{definition}[theorem]{Definition}

\newtheorem{remark}[theorem]{Remark}

\newcommand{\D}{{\mathcal D}}

\newcommand{\G}{{\mathcal G}}

\newcommand{\I}{{\mathcal I}}

\newcommand{\M}{{\mathcal M}}

\newcommand{\bci}{\bC^{i}}

\newcommand{\bM}{{\boldsymbol{M}}}

\newcommand{\bF}{{\boldsymbol F}} 
\newcommand{\bH}{{\boldsymbol H}} 
\newcommand{\bS}{{\boldsymbol S}} 
\newcommand{\bT}{{\boldsymbol T}}

\newcommand{\bB}{{\boldsymbol B}}
\newcommand{\bC}{{\boldsymbol C}}
\newcommand{\bE}{{\boldsymbol E}}

\newcommand{\bb}{{\boldsymbol b}}
\newcommand{\bu}{{\boldsymbol u}}
\newcommand{\bv}{{\boldsymbol v}}
\newcommand{\bff}{{\boldsymbol f}}

\newcommand{\by}{{\boldsymbol y}}

\newcommand{\bc}{{\boldsymbol c}}
\newcommand{\bO}{{\boldsymbol 0}}

\newcommand{\bx}{{\boldsymbol{x}}}

\newcommand{\bG}{{\boldsymbol{G}}}
\newcommand{\be}{{\boldsymbol e}}

\newcommand{\ee}{\overline{\boldsymbol{E}}}
\newcommand{\he}{\overline{\boldsymbol{H}}}
\newcommand{\ce}{\boldsymbol{c}^{\boldsymbol{E}}}
\newcommand{\cee}{\boldsymbol{c}^{\overline{\boldsymbol{E}}}}
\newcommand{\cj}{\boldsymbol{c}^{j}}

\newcommand{\ff}{\mathbb{F}}

\newcommand{\fq}{\mathbb{F}_q}

\newcommand{\al}{\alpha}
\newcommand{\bt}{\beta}

\newcommand{\lam}{\lambda}

\newcommand{\mi}{{\sf I}}

\newcommand{\dist}{{\mathsf{d}}}

\newcommand{\cC}{{\mathscr{C}}}

\newcommand{\cce}{\mathscr{C}^{\boldsymbol{E}}}
\newcommand{\ccee}{\mathscr{C}^{\overline{\boldsymbol{E}}}}
\newcommand{\ccj}{\mathscr{C}^{j}}

\newcommand{\define}{\stackrel{\mbox{\tiny $\triangle$}}{=}}

\newcommand{\nin}{\noindent}

\newcommand{\seq}{\subseteq}

\newcommand{\et}{{\emph{et al.}}}

\begin{document}
\pagestyle{plain}

\title{On Block Security of Regenerating Codes at the MBR Point for Distributed Storage Systems}

\author{
   \IEEEauthorblockN{
     Son Hoang Dau\IEEEauthorrefmark{1},
		 Wentu Song\IEEEauthorrefmark{2},
     Chau Yuen\IEEEauthorrefmark{3}}
   \IEEEauthorblockA{
     Singapore University of Technology and Design, Singapore\\     
		Emails: $\{$\IEEEauthorrefmark{1}{\it sonhoang\_dau},
		\IEEEauthorrefmark{2}{\it wentu\_song},		
	   \IEEEauthorrefmark{3}{\it yuenchau}$\}$@sutd.edu.sg
 }
}
\maketitle

\begin{abstract}
A \emph{passive} adversary can eavesdrop stored content or downloaded content of some storage nodes, 
in order to learn illegally about the file stored across a distributed storage system
(DSS). 
Previous work in the literature focuses on code constructions that trade storage capacity for \emph{perfect security}. 
In other words, by decreasing the amount of original data that it can store, the system can guarantee that the adversary, 
which eavesdrops up to a certain number of
storage nodes, obtains \emph{no information} (in Shannon's sense) about the original data.     
In this work we introduce the concept of \emph{block security} for DSS and investigate 
minimum bandwidth regenerating (MBR) codes that are \emph{block secure} against
adversaries of varied eavesdropping strengths. Such MBR codes guarantee that no information about any \emph{group}
of original data units up to a certain size is revealed, \emph{without} sacrificing the storage capacity of the system. 
The size of such secure groups varies according to the number of nodes that the adversary can eavesdrop. We show that code constructions based on Cauchy matrices provide block security. The opposite conclusion is drawn for codes based on Vandermonde matrices. 
\end{abstract}

\section{Introduction}
\label{sec:intro}

\subsection{Background}
\label{subsec:background}

In recent years, the demand for large-scale data storage has grown explosively, due to numerous applications including large files and video sharing, social networks, and back-up systems. Distributed Storage Systems (DSS) store a tremendous amount of data using a massive collection of distributed storage nodes. Applications of DSS include large data centers and P2P storage systems such as OceanStore~\cite{Rhea-et2006}, Total Recall~\cite{Bhagwan-et2004}, and Dhash++~\cite{Dabek-et2004} that deploy a huge number of storage nodes spread widely over the Internet. Since any storage device is individually unreliable and subject to failure, redundancy must be introduced to provide the much-needed system-level protection against data loss due to storage node failure.

The simplest form of redundancy is replication. By storing $n$ identical copies of a file at $n$ distributed nodes, one copy per node, an $n$-duplication system can guarantee the data availability as long as no more than $(n -1)$ nodes fail. Such systems are very easy to implement and maintain, but extremely inefficient in storage space utilization, incurring tremendous waste in equipment, building space, and cost for powering and cooling. More sophisticated systems employing erasure coding based on maximum distance separable (MDS) codes~\cite{WeatherspoonKubiatowicz2001} can expect to considerably improve the storage efficiency. However, a DSS based on MDS codes often incurs considerable communication bandwidth during node repair: a replacement node has to download the whole file to recover the content of just one node. To overcome the disadvantages of both replication and erasure codes, regenerating codes were proposed for DSS in the groundbreaking work of Dimakis {\et}~\cite{DimakisGodfreyWainwrightRamchandran2007}. 
In this work, we only consider regenerating codes for DSS and limit ourselves to 
minimum bandwidth regenerating (MBR) codes with exact repair~\cite{DimakisGodfreyWainwrightRamchandran2007}. 

In practice, one important aspect in the design of a DSS code is its security. 
As the storage nodes can well be located at different places, and new nodes 
join the network all the time to replace failing nodes, it is possible that 
at a certain point, some nodes might be compromised by some unauthorized party, 
referred to as an adversary. The adversary can be either 
\begin{itemize}
\item
\emph{active}, i.e. controlling the node, 
modifying the data, and sending erroneous data to other nodes 
\cite{PawarRouayhebRamchandran2011_1,PawarRouayhebRamchandran2011_2,
RashmiShahKumar2012,DikaliotisDimakisHo2010}, or 
\item
\emph{passive}, 
i.e. knowing the data stored in the node and observing all communication between 
this node and the other nodes, in order to learn illegally the content of the file stored 
by the DSS \cite{PawarRouayhebRamchandran2011_2, PawarRouayhebRamchandran2010,ShahRashmiKumar2011, RawatKoyluogluSilbersteinVishwanath2013, GoparajuRouayhebCalderbankPoor2013}. 
\end{itemize}
In this work, we focus on the security of DSS codes against 
the second type of adversary, the passive adversary. We also refer to this type of adversary 
as eavesdropper. We illustrate in Fig.~\ref{fig:ex1} a scenario where the adversary can observe all data stored in one storage node. 
The data file is split into five data units $\bff = (f_1,f_2,f_3,f_4,f_5)$. 
Originally there are four storage nodes. At a certain time, Node 4 fails. 
Node 5 comes in to replace Node 4 (repair). 
If the content of this node is eavesdropped by Eve, then Eve can obtain
the values of two units $f_4$ and $f_5$ explicitly. 
The same thing happens when Eve eavesdrops any node among the three nodes 1, 2, 
and 3. 
Therefore the system is not secure against an adversary who can eavesdrop
one storage node. 
\begin{figure}
\centering
\includegraphics[scale=0.65]{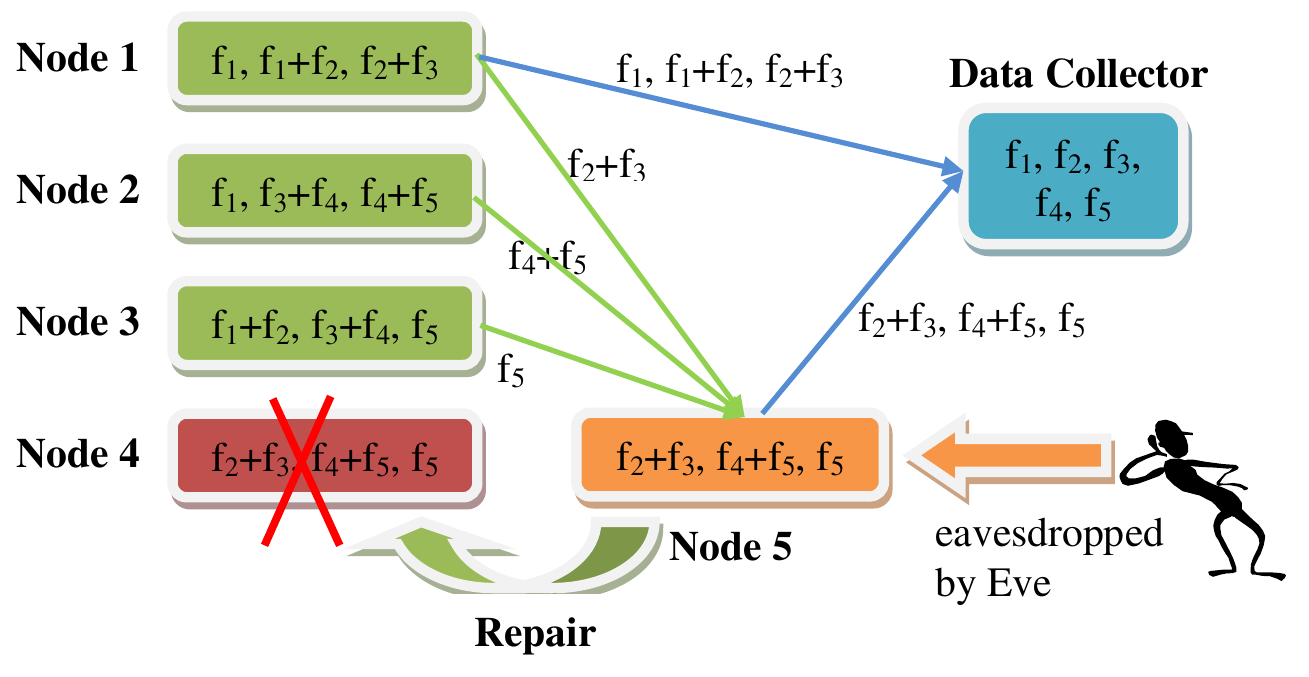}
\caption{Example of a DSS with an eavesdropper.}
\label{fig:ex1}
\vspace{-20pt}
\end{figure}
\subsection{Related Work}
\label{subsec:related_work}

Hereafter we follow the standard notation in the regenerating code literature
(see Section~\ref{sec:preliminaries}).  
In their pioneering work, Dimakis {\et}~\cite{DimakisGodfreyWainwrightRamchandran2007}
established that the maximum file size to be stored in a DSS $\D(n,k,d)$ must 
satisfy the following inequality
\begin{equation} 
\label{eq:dimakis}
\vspace{-3pt}
\M \leq \sum_{i = 1}^k \min\{(d-i+1)\bt, \al\}. 
\end{equation} 
A construction of optimal regenerating codes with exact repair at the MBR point for all $n, k, d$ was proposed by Rashmi {\et}~\cite{RashmiShahKumar2011}.   
When the stored contents of some $\ell$ nodes are observed by an adversary Eve, 
Pawar {\et}~\cite{PawarRouayhebRamchandran2011_1,PawarRouayhebRamchandran2011_2}
showed that the maximum file size to be stored satisfies
\begin{equation} 
\label{eq:pawar_bound}
\M^{\text{s}} \leq \sum_{i = \ell+1}^k \min\{(d-i+1)\bt, \al\},
\end{equation} 
provided that Eve gains \emph{no information} (in Shannon's sense) about the file.
This type of security is often referred to as \emph{perfect security} or \emph{strong security} in the literature.  
The parameter $\ell$ is called the \emph{adversarial strength}.   
The authors~\cite{PawarRouayhebRamchandran2011_1,PawarRouayhebRamchandran2011_2} also provided an optimal code construction, based on complete graphs, for the case $d = n-1$ that
attains the bound (\ref{eq:pawar_bound}). We later argue that an extension 
of their construction based on regular graphs~\cite{RouayhebRamchandran2010}
also produces optimal perfectly secure codes for all $d$ with $nd$ even (see Remark~\ref{rm:fractional}). 
Under the same adversary model, Shah {\et}~\cite{ShahRashmiKumar2011}
constructed optimal codes that attain the bound (\ref{eq:pawar_bound}) for all
$d$. The authors used product-matrix codes in their construction.

Comparing (\ref{eq:dimakis}) and (\ref{eq:pawar_bound}), 
it is clear that in order to achieve the perfect security against an adversary
which eavesdrops $\ell$ storage nodes, the storage capacity of the system 
has to be decreased by an amount of
\[
\vspace{-3pt}
\Delta \M = \M - \M^{\text{s}} = \sum_{i = 1}^{\ell} \min\{(d-i+1)\bt, \al\}. 
\]  
Therefore, perfect security is obtained at the cost of lowering the storage capacity. 
Additionally, according to the security scheme studied thereof, one has to specify 
beforehand the strength of the adversary, i.e. the number of storage nodes 
it can eavesdrop, and then modify the input accordingly to obtain the perfect security. 
When the strength of the adversary exceeds the specified threshold, 
nothing is guaranteed on the security of the system.
In many practical storage systems, perfect security is either too strict 
and might not be even necessary, or too costly (too much wasted space in the system) 
and might not be affordable.  
Hence weaker security levels at lower costs would be preferred in many practical scenarios.

\subsection{Our Contribution}
\label{subsec:our_contribution}

In order to address the aforementioned issues with perfect security, 
we propose to use a broader concept of security from the network coding
literature, namely, \emph{block security}~\cite{BhattadNarayanan2005,
SilvaKschischang2009, DauSkachekChee2012} 
(also known as \emph{security against guessing}). 
A code is $b$-block secure against an adversary of strength $\ell$, if the adversary, which accesses at most $\ell$ storage nodes, gains no information about any group of 
$b$ data units. 
In the language of information theory, the mutual information between the eavesdropped content and any group of $b$ original data units is zero. 
The concept of block security describes nicely a hierarchy of different levels of security against eavesdropping. Depending on the adversarial strength, 
a block secure system can provide a range of different levels of security, from the weakest level to the strongest level: 
\begin{itemize}
\item	$1$-block security (also referred to as \emph{weak security)} implies that no individual data unit is revealed, 
\item	$2$-block security implies that no information on any group of two data units is revealed, 
\item	$\cdots$,
\item $\M$-block security, where $\M$ is the total number of data units (the file size), implies the perfect security, i.e. no information about the original data is revealed.  
\end{itemize}
If the system is weakly secure, although the adversary gains some information
about the file, it cannot determine each particular data unit. For instance, if the
file is a movie and the data units are movie chunks, then the adversary obtains no 
information about each individual chunk, and hence cannot play the movie. 
Furthermore, if the system is $b$-block secure against an adversary of strength $\ell$, even when the adversary can access some $\ell$ storage nodes
and on top of that, gain knowledge on some other $b-1$ data units via some side channel, it still cannot determine each particular data unit.  
 
We investigate the security of the two main existing MBR code constructions
\cite{RashmiShahKumarRamchandran2009, RouayhebRamchandran2010,  RashmiShahKumar2011}. Recall that either Vandermonde matrices or 
Cauchy matrices are often used in these constructions. We prove the following 
\begin{itemize}
\item
Vandermonde-based codes are \emph{not} block secure, 
\item 
Cauchy-based regular-graph codes~\cite{RashmiShahKumarRamchandran2009, RouayhebRamchandran2010} are inherently block secure,
\item 
Cauchy-based product-matrix codes~\cite{RashmiShahKumar2011} are inherently block secure. 
\end{itemize}
Based on these results, we are able to conclude that using Cauchy matrices rather than Vandermonde in these constructions of MBR
codes automatically guarantees weaker levels of security compared with 
perfect security, but most importantly, \emph{without} any loss on the storage capacity
of the system. We also show that with Cauchy-based constructions of perfectly secure codes~\cite{PawarRouayhebRamchandran2010, 
PawarRouayhebRamchandran2011_2, ShahRashmiKumar2011}, 
certain security level is still guaranteed even when the adversarial 
strength \emph{surpasses} the security threshold.  
More specifically, in addition to being \emph{perfectly secure} against adversaries of strength not exceeding
a specified threshold $\lambda$, the codes remain to be \emph{block secure} against an
adversary of strength $\ell > \lambda$ as long as $\ell < k$   
(Fig.~\ref{fig:sec_level}). 
The level $b$ of block security gradually decreases as the adversarial strength $\ell$ increases. When $\ell \geq k$, the whole file is revealed. 
\begin{figure}[H]
\vspace{-5pt}
\centering
\includegraphics[scale=1]{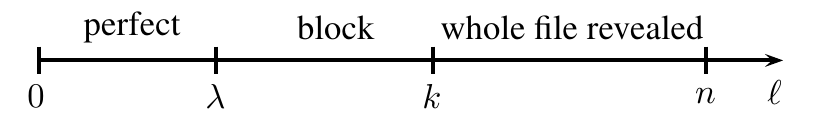}
\caption{Security of a Cauchy-based system as the adversarial strength $\ell$ varies.}
\label{fig:sec_level}
\vspace{-5pt}
\end{figure}

The paper is organized as follows. In Section~\ref{sec:preliminaries}, we discuss the concept of block security for distributed storage systems. We analyze the security of regular-graph codes and product-matrix codes 
in Section~\ref{sec:gcode} and Section~\ref{sec:pcode}, respectively. 
We conclude the paper in Section~\ref{sec:conclusion}. 

\section{Block Security for Distributed Storage Systems}
\label{sec:preliminaries}

We denote by $\fq$ the finite field with $q$ elements. 
Let $[n]$ denote the set $\{1,2,\ldots,n\}$. 
Let $\bu^{\text{t}}$ denotes the transpose of a vector $\bu$. 
For an $m \times n$ matrix $\bM$, for $i \in [m]$ and $j \in [n]$, let
$\bM_i$ and $\bM[j]$ denote the row $i$ and the column $j$ of $\bM$, respectively. 
We define below standard notions from coding theory (for instance, see \cite{MW_S}).
 
The (Hamming) \emph{weight} of a vector $\bu \in \fq^n$ is the number
of the nonzero coordinates of $\bu$. 
For the vectors $\bu = (u_1, u_2, \ldots, u_n) \in \fq^n$ and $\bv = (v_1, v_2, \ldots, v_n) \in \fq^n$, 
the (Hamming) distance between $\bu$ and $\bv$ is defined to be the number of coordinates 
where $\bu$ and $\bv$ differ, namely, 
\[
\dist(\bu,\bv) = |\{i \in [n] \; : \; u_i \ne v_i\}|. 
\]
A $k$-dimensional subspace $\cC$ of $\fq^n$ is called a linear $[n,k,d]_q$ code over $\fq$ if the minimum distance of $\cC$, 
\[
\dist(\cC) \define \min_{\bu \in \cC, \; \bv \in \cC, \; \bu \neq \bv} \dist(\bu,\bv) \; , 
\]
is equal to $d$. Sometimes we may use the notation $[n,k]_q$ or just $[n,k]$ for the sake of simplicity. The vectors in $\cC$ are called codewords. It is easy to see that the minimum weight of a nonzero codeword in a linear code $\cC$ is equal to its minimum distance $\dist(\cC)$. A \emph{generator matrix} $\bG$ of an $[n,k]_q$ code $\cC$ is a $k \times n$ matrix whose rows are linearly independent codewords of $\cC$. Then $\cC = \{\by \bG: \by \in \fq^k\}$. 
The well-known Singleton bound states that for any $[n,k,d]_q$
code, it holds that $d \leq n - k + 1$. If the equality is attained, the code is called maximum distance separable (MDS).  

Let $\bF = (F_1, F_2, \ldots, F_\M)$, where $F_i$'s $(i \in [\M])$ are 
independent and identically uniformly distributed random variables over $\fq$.  
We assume that the file to be stored in the system is $\bff = (f_1, f_2, \ldots, f_\M) \in \fq^\M$, a realization of $\bF$.  
We call $\M$ the file size and each $f_i$ a (original) data unit.  

We denote by $\D(n,k,d)$ a typical DSS with $n$ storage nodes
where the file can be recovered from the contents of any $k$ out of $n$ nodes, 
and to repair a failed node, a new node (referred to as a newcomer)
can contact any $d$ nodes to regenerate the content of the failed node. 
We refer to $d$ as the \emph{repair degree}. Additionally, each node stores
$\al$ coded units, i.e. $\al$ linear combinations of data units $f_i$'s. 
In the repair process, a newcomer downloads $\bt$ coded units from each of $d$ 
live nodes.
Suppose that $\bt = 1$ (data striping is used for larger $\bt$). 
At the MBR point, we have~\cite{RashmiShahKumarRamchandran2009,
RashmiShahKumar2011}
\[
\M = kd - \binom{k}{2}, \quad \al = d. 
\]
Following the adversarial attack model proposed by Pawar {\et}~\cite{PawarRouayhebRamchandran2010, PawarRouayhebRamchandran2011_2}, 
we suppose that the adversary can observe the stored contents of some $\ell$ nodes. 
Let $\bci$ denote the random vector over $\fq^{\al}$ 
that represents the stored content at Node $i$. 

\vskip 10pt
\begin{definition}[\cite{PawarRouayhebRamchandran2010, PawarRouayhebRamchandran2011_2}]
\label{def:perfect_security}
A DSS $\D(n,k,d)$ together with its coding scheme is called \emph{perfectly secure}
against an adversary of strength $\lam$ $(\lam < k)$ if the mutual information
\[
\mi(\bF,\cup_{i \in E} \bC^i) = 0,
\] 
for all subsets $E \subseteq [n]$, $|E| \leq \lam$. 
\end{definition}

\vskip 10pt
\begin{definition}
\label{def:block_security}
A DSS $\D(n,k,d)$ together with its coding scheme is called $b$-\emph{block secure}
against an adversary of strength $\lam$ $(\lam < k)$ if the mutual information
\[
\mi(\cup_{j \in B} F_j,\cup_{i \in E} \bC^i) = 0,
\] 
for all subsets $B \seq [\M]$, $|B| \leq b$, and for all subsets $E \subseteq [n]$, $|E| \leq \lam$. 
\end{definition}

\begin{figure}[H]
\centering
\includegraphics[scale=0.55]{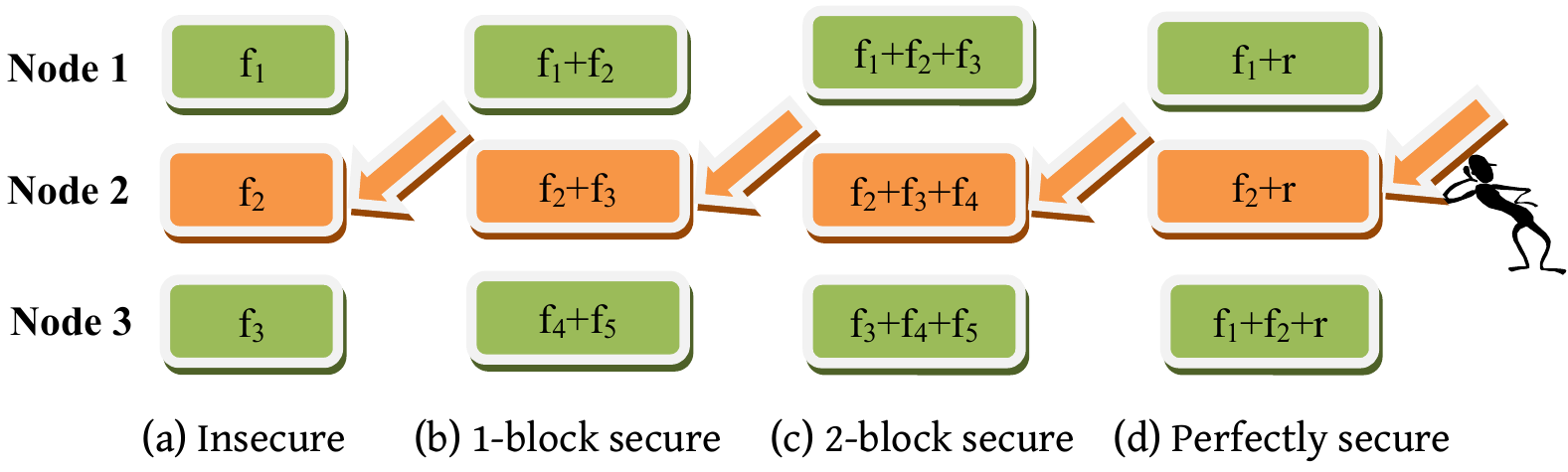}
\caption{Different levels of security.}
\label{fig:levels_security}
\vspace{-10pt}
\end{figure}

We illustrate in Fig.~\ref{fig:levels_security} different levels of security for DSS. 
For instance, in Fig.~\ref{fig:levels_security}(c), as long as the adversary 
accesses the stored content of only one node, it cannot deduce any linear combination
of two units. Therefore, the system is 2-block secure against an adversary of
strength one. In Fig.~\ref{fig:levels_security}(d), by using a randomly generated
variable $r$, as long as the adversary accesses the stored content of only one node, 
it cannot deduce any linear combination of the data units, and hence gains no information at all about the stored file. Hence in this case, the system is perfectly
secure against an adversary of strength one.  

It is straightforward that $\M$-block security is equivalent to perfect security, 
where $\M$ is the file size. Hence, perfect security can well be regarded as the highest
level of block security. However, perfect security is not given for free. 
One has to trade some storage capacity for perfect security (see Section~\ref{subsec:related_work}). 
In fact, in the perfectly secure code constructions presented in~\cite{PawarRouayhebRamchandran2010, PawarRouayhebRamchandran2011_2, 
ShahRashmiKumar2011}, part of the file has to be replaced by randomly 
generated variables, hence reducing the useful storage space of the system.  
By contrast, lower levels of block security can be achieved at essentially no cost, 
as we later present in Section~\ref{sec:gcode} and~\ref{sec:pcode}. 
This advantage of block security makes the concept attractive to practical storage
systems, where the storage redundancy has to be minimized to maintain
a competitive price for the storage service.   

Under the assumptions that 
\begin{itemize}
\item the data units are all independent and identically 
uniformly distributed random variables over $\fq$, 
\item the coding scheme is linear, 
\end{itemize}
the $b$-block security given in Definition~\ref{def:block_security} 
is equivalent to the requirement that no linear combination
of at most $b$ data units can be deduced by the adversary.   
The following lemma specifies a necessary and sufficient condition for
the block security of DSS. 

\vskip 10pt 
\begin{lemma} 
\label{lem:basic}
Let $\bff = (f_1,f_2,\ldots,f_\M) \in \fq^{\M}$ be the stored file and $\bE \bff^{\text{t}}$ represent the
coded units that the adversary obtains by observing $\ell$ storage nodes. 
Let $\cC^{\bE}$ be the linear error-correcting code generated by the rows of the matrix $\bE$, 
and $\dist(\cC^{\bE})$ be its minimum distance. 
Then the adversary cannot deduce any nontrivial linear combination of any group of at most
$b$ data units if and only if $b \leq \dist(\cC^{\bE})-1$.  
\end{lemma}  
\vskip 10pt 
A rigorous proof of Lemma~\ref{lem:basic} can be found in the Appendix. We discuss below the 
intuition behind the proof of this lemma. 
As the adversary obtains $\bE \bff^{\text{t}}$, it can linearly transform these coded symbols
by considering the product $\boldsymbol{\alpha}\bE \bff^{\text{t}}$, where $\boldsymbol{\alpha} 
$ is some coefficient vector. Since $\boldsymbol{\alpha}\bE$
is a codeword of $\cC^{\bE}$, its weight is at least $\dist(\cC^{\bE})$ if it is a nonzero codeword. 
In other words, if $\boldsymbol{\alpha}\bE \neq \bO$ then it has at least $\dist(\cC^{\bE})$
nonzero coordinates. As a result, $\boldsymbol{\alpha}\bE \bff^{\text{t}}$ is a linear combination
of at least $\dist(\cC^{\bE})$ data units. 
Therefore, by linearly transforming the eavesdropped coded units $\bE \bff^{\text{t}}$,
the adversary cannot produce a nontrivial linear combination of $\dist(\cC^{\bE})-1$ data units or less. 

On the other hand, as the adversary can choose an appropriate vector $\boldsymbol{\alpha}$
so that $\boldsymbol{\alpha}\bE$ has weight exactly $\dist(\cC^{\bE})$,
it can always determine the value of a linear combination of a certain group of $\dist(\cC^{\bE})$ data units.  
Thus, the adversary cannot deduce any nontrivial linear combination of any group of at most
$b$ data units \emph{only if} $b \leq \dist(\cC^{\bE})-1$.  

\section{On the Security of Regular-Graph Codes}
\label{sec:gcode}

We briefly describe the regular-graph codes constructed by Rashmi {\et}~\cite{RashmiShahKumarRamchandran2009}
and El Rouayheb {\et}~\cite{RouayhebRamchandran2010} below.
 
Let $\G$ be a $d$-regular graph on $n$ vertices $u_1,u_2,\ldots, u_n$. 
Then each vertex of $\G$ is adjacent to $d$ edges. Therefore, let $e_1,e_2,\ldots,e_{nd/2}$
be all $nd/2$ edges of $\G$. Let $\bG$ be an $(nd/2) \times \M$ matrix over $\fq$
satisfying the MDS property: any $\M$ rows of $\bG$ are linearly independent. 
Then $\M$ data units $f_1, f_2,\ldots,f_\M$ are encoded into $nd/2$ coded units by the
transformation $\bc^{\text{t}} = (c_1,\ldots,c_{nd/2})= \bG \bff^{\text{t}}$. 
Node $i$ stores $c_j$ if and only if the edge $e_j$ is adjacent to the vertex $u_i$. 
As $\G$ is $d$-regular, each node stores exactly $\al = d$ coded units $c_j$'s. 
Any set of $k$ nodes together provide $kd - \binom{k}{2} = \M$ distinct coded 
units $c_j$'s, hence can recover the whole file thank to the MDS property of 
the encoding matrix $\bG$. When Node $i$ fails $(i \in [n])$, 
the newcomer contacts Node $j$ if $u_i$ and $u_j$ are adjacent in $\G$. 
Since $\G$ is $d$-regular, there are $d$ such nodes. 
For such a Node $j$, the newcomer downloads $\bt = 1$ coded unit $c_s$
if $e_s$ is the edge connected $u_i$ and $u_j$. 
These $d$ coded units $c_s$ are distinct as they correspond to $d$ distinct edges
adjacent to $u_i$, and form the content of the failed node. 
The newcomer simply stores these $d$ coded units and the repair process for Node $i$ is done.   

\subsection{Cauchy-Based Regular-Graph Codes Are Block Secure}
\label{subsec:cgcode}

A Cauchy matrix 
is a matrix of form $\Big((x_i+y_j)^{-1}\Big)_{m\times n}$, where
$x_i$'s and $y_j$'s are elements of $\fq$ that satisfy $x_i + y_j \neq 0$
for all $i \in [m]$ and $j \in [n]$. A Cauchy matrix has a special property that 
any submatrix is again a Cauchy matrix. 
It is well known that any square Cauchy matrix is invertible. 
Therefore, any square submatrix of a Cauchy matrix is invertible. 
This is a crucial property that makes codes based on Cauchy matrices block secure. 

 \vskip 10pt 
\begin{theorem} 
\label{thm:cauchy_gcode}
A DSS $\D(n,k,d)$ equipped with a regular-graph regenerating code~\cite{RouayhebRamchandran2010} 
based on a Cauchy matrix is $b$-block secure against an adversary of strength $\ell$ $(\ell < k)$, 
where 
\[
b = \Big( kd - \binom{k}{2} \Big) - \Big( \ell d - \binom{\ell}{2} \Big)= (k - \ell)\Big( d - \dfrac{k+\ell-1}{2} \Big). 
\] 
Note that $nd$ must be even according to the code construction. 
\end{theorem} 
\begin{proof} 
Note that according to the code construction~\cite{RashmiShahKumarRamchandran2009, RouayhebRamchandran2010}, an adversary that accesses $\ell$ nodes
obtains $\ell d - \binom{\ell}{2}$ distinct coded units. 
Therefore, the adversary obtains $\bE \bff^{\text{t}}$, where $\bE$
is a submatrix of the encoding matrix $\bG$, consisting of 
$\ell d - \binom{\ell}{2}$ rows of $\bG$. 
Note that $\bE$ has $kd - \binom{k}{2}$ columns. 
As $\bG$ is a Cauchy matrix, any square submatrix of size $\ell d - \binom{\ell}{2}$ of $\bE$ is also a Cauchy matrix, and hence invertible. 
Therefore, the rows of $\bE$ generates an MDS code of length 
$kd - \binom{k}{2}$ and dimension $\ell d - \binom{\ell}{2}$~\cite[Ch. 11]{MW_S}. It is also well known that such an MDS code has 
minimum distance $(kd - \binom{k}{2}) - (\ell d - \binom{\ell}{2}) + 1$. 
Applying Lemma~\ref{lem:basic}, the proof follows. 
\end{proof} 
\vskip 10pt 

For instance, according to Theorem~\ref{thm:cauchy_gcode}, for $n = 7$, $k = 5$, $d = 6$, the degradation of the block security level of 
a Cauchy-based regular-graph code is illustrated in Fig.~\ref{fig:block_security_level_graph}. 
Note that the file size is $\M = 5 \times 6 - \binom{5}{2} = 20$.

\begin{figure}[htb]
\centering
\includegraphics{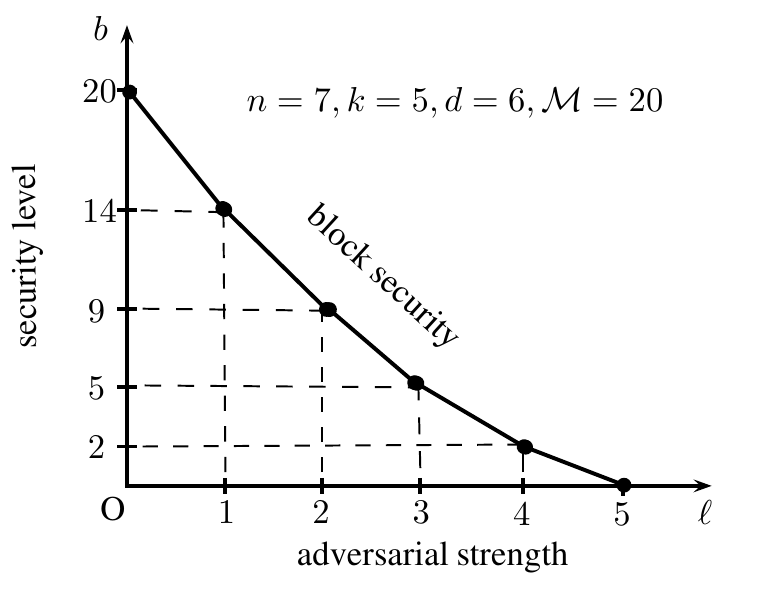}
\caption{Degradation of the security level for $\D(7, 5, 6)$ with a Cauchy-based regular-graph code as the adversarial strength increases.}
\label{fig:block_security_level_graph}
\end{figure}

The corresponding construction of perfectly secure codes~\cite{PawarRouayhebRamchandran2010, PawarRouayhebRamchandran2011_2}
is completely the same as the one described at the beginning of this section, except that 
a subset of data units has to be replaced by a subset of randomly generated variables and
that the encoding matrix $\bG$ has to be either Vandermonde or Cauchy (or more generally, 
a (transposed) generator matrix of a nested MDS code).   
Suppose that the Cauchy matrix is used in this construction. 
Treating the random variables as data units, we can apply Theorem~\ref{thm:cauchy_gcode} 
and show that even when the adversarial strength surpasses the specified threshold, 
although the code is no longer perfectly secure, it is still block secure.
Hence we have the following corollary. 

\vskip 10pt 
\begin{corollary}
\label{cr:gcode}
Consider a DSS $\D(n,k,d)$ equipped with the regenerating code
proposed in~\cite{PawarRouayhebRamchandran2010, PawarRouayhebRamchandran2011_2, RouayhebRamchandran2010}, 
which is perfectly secure against an adversary of strength at most $\lam$. 
Suppose that in the construction of the regenerating code, a Cauchy matrix is used. 
Then when the adversarial strength $\ell$ $(\ell < k)$ surpasses the threshold $\lam$, 
the code is no longer perfectly secure, but is still $b$-block secure, where 
\[
b = \Big( kd - \binom{k}{2} \Big) - \Big( \ell d - \binom{\ell}{2} \Big) = (k - \ell)\Big( d - \dfrac{k+\ell-1}{2} \Big).
\] 
\end{corollary}
\vskip 5pt 

For example, for a DSS $\D(7, 5, 6)$ as in Fig.~\ref{fig:block_security_level_graph}, 
suppose that the specified security threshold
is $\lambda = 2$, then the degradation of the security is depicted in Fig.~\ref{fig:perfect_block}. 
Note that now the file size has to be decreased to $\M^{\text{s}} = 9$, according to (\ref{eq:pawar_bound}).    

\begin{figure}[htb]
\centering
\includegraphics{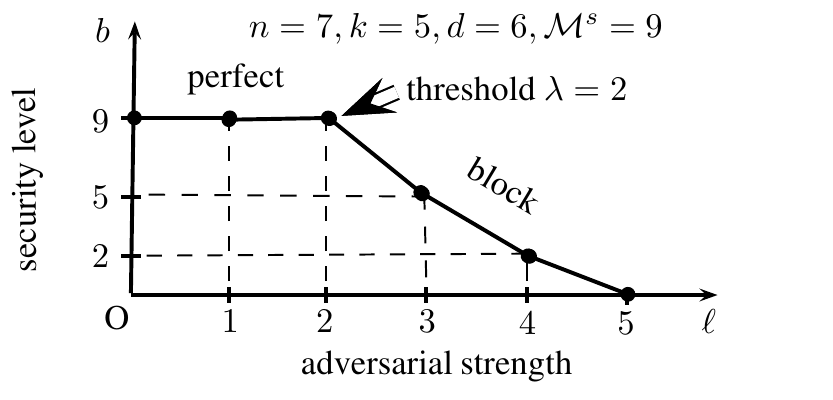}
\caption{Degradation of the security level for $\D(7, 5, 6)$ with a perfectly secure Cauchy-based regular-graph code 
as the adversarial strength increases.}
\label{fig:perfect_block}
\end{figure}
\vspace{-5pt}

\begin{remark}
\label{rm:fractional}
The construction of optimal perfectly secure codes with uncoded exact repair was
proposed by Pawar {\et}~\cite{PawarRouayhebRamchandran2010, PawarRouayhebRamchandran2011_2} for the case $d = n - 1$. 
We argue below that the same construction, i.e. replacing data units by random units and using 
nested MDS codes, based on regular-graphs~\cite{RouayhebRamchandran2010} also provides optimal perfectly secure codes
for all $d$ with $nd$ even. 
That is why we mention both constructions in Corollary~\ref{cr:gcode}. 
Indeed, as $\bt = 1$ and $\al = d$, the bound (\ref{eq:pawar_bound}) reduces to 
\begin{equation} 
\label{eq:reduced_Pawar}
\M^{\text{s}} \leq (k - \ell)\big( d - (k+\ell-1)/2 \big). 
\end{equation} 
According to the construction in~\cite{RouayhebRamchandran2010}, 
eavesdropping on $\ell$ nodes reveals $R = \ell d - \binom{\ell}{2}$ 
distinct coded units. 
Hence, $\ell d - \binom{\ell}{2}$ data units has to be replaced by random units. 
Therefore, the number of data units that can be 
stored securely using a nested MDS code is
\[
\begin{split} 
\M^{\text{s}} &= \M - R\\
&= \Big(kd - \binom{k}{2}\Big) - \Big( \ell d - \binom{\ell}{2} \Big)\\
&= (k - \ell)\big( d - (k+\ell-1)/2 \big), 
\end{split}
\]
which matches the bound (\ref{eq:reduced_Pawar}). Hence the 
regular-graph codes~\cite{RouayhebRamchandran2010} are optimal perfectly
secure code for every $d$ with $nd$ even.  
\end{remark}

\subsection{Vandermonde-Based Regular-Graph Codes Are Not Block Secure}
\label{subsec:vgcode}

We present below an example of a Vandermonde-based regular-graph code~\cite{RashmiShahKumarRamchandran2009, RouayhebRamchandran2010}
that is not block secure and briefly explain the reason behind it.

Let $n = 4$, $k = 2$, and $d = 3$. As $\bt = 1$, we have $\M = kd -\binom{k}{2} = 5$. 
The encoding matrix $\bG$ is chosen as a $6 \times 5$ Vandermonde matrix over $\ff_{13}$
\[
\bG = 
\begin{pmatrix}
1& 1& 1& 1& 1\\
1& 3& 9& 1& 3\\
1& 5& 12& 8& 1\\
1& 7& 10& 5& 9\\
1& 9 & 3 & 1 & 9\\ 
1& 11& 4 & 5& 3\\
\end{pmatrix}
\]  
An adversary which accesses one node obtains $d = 3$ distinct coded units. 
Suppose these three coded units are $\bG_1 \bff^{\text{t}}$, $\bG_2 \bff^{\text{t}}$, and $\bG_5 \bff^{\text{t}}$.
In other words, the adversary obtains $\bE \bff^{\text{t}}$ where $\bE$ is the submatrix of
$\bG$ consisting of the first, the second, and the fifth rows of $\bG$. 
It is straightforward to verify that the minimum distance of the error-correcting code generated
by the rows of $\bE$ is one. Therefore, according to Lemma~\ref{lem:basic}, the code is not even weakly secure ($1$-block secure) against an adversary of strength one. 
More specifically, by applying a linear transformation $\bv \bE \bff^{\text{t}}$ where
$\bv = (9, 1, 3)$, the adversary obtains $f_3 = \bv \bE \bff^{\text{t}}$ explicitly.  

The reason behind this unwanted behavior of Vandermonde-based codes
can be explained as follows. The block security of Cauchy-based regular-graph
codes (see Proof of Theorem~\ref{thm:cauchy_gcode}) strictly relies on a very special property of a Cauchy matrix: every square submatrix
of a Cauchy matrix is invertible. However, a Vandermonde matrix does not have this 
property. For example, the $3 \times 3$ submatrix of the Vandermonde 
matrix $\bG$ above that consisting of the entries in the first, the second, and the fifth
rows, and in the first, second, and the fifth columns of $\G$ only has rank two. 
Therefore, the block security of the corresponding code is no longer guaranteed.  

\section{On the Security of Product-Matrix Codes}
\label{sec:pcode}

We briefly describe the MBR product-matrix codes constructed by Rashmi {\et}~\cite{RashmiShahKumar2011} below.

The file size $\M = kd - \binom{k}{2}$ can be rewritten as $\M = \binom{k+1}{2} - k(d-k)$. Let $\bM$ be the \emph{message} matrix of the following form
\[
\vspace{-3pt}
\bM = \begin{pmatrix}
\bS & \bT\\
\bT^{\text{t}} & \bO
\end{pmatrix},
\vspace{-3pt}
\]
where $\bS$ is a $k\times k$ symmetric matrix and $\bT$ is a $k \times (d-k)$ 
matrix. The $\binom{k+1}{2}$ entries in the upper-triangular half of $\bS$ are
filled up by $\binom{k+1}{2}$ distinct data units drawn from the set $\{f_i\}_{i \in [\M]}$. The remaining $k(d-k)$ data units are used to fill up the second 
$k \times (d-k)$ matrix $\bT$. The encoding matrix is $\Psi = [\Phi \quad \Delta]$, where 
$\Phi$ and $\Delta$ are $n \times k$ and $n \times (d-k)$ matrices, respectively, 
chosen in such a way that
\begin{enumerate}
	\item any $d$ rows of $\Psi$ are linearly independent, 
	\item any $k$ rows of $\Phi$ are linearly independent. 
\end{enumerate}
Then Node $i$ stores $d$ entries of row $i$ of the matrix $\Psi \bM$. 
The encoding matrix $\Psi$ is often chosen as a Vandermonde matrix 
or a Cauchy matrix. 
Details on the file reconstruction and node repair processes can be found 
in~\cite{RashmiShahKumar2011}.   

\subsection{Cauchy-Based Product-Matrix Codes Are Block Secure}
\label{subsec:cpcode}

Our main result in this section is to prove that the Cauchy-based product-matrix
code~\cite{RashmiShahKumar2011} is $(k-\ell)$-block secure against an adversary 
of strength at most $\ell < k$. Note that if $k$ nodes are eavesdropped, the 
whole file will be reconstructed. Thus the block security level of the Cauchy-based
product-matrix codes is $d - (k+\ell-1)/2$ times lower than that of the Cauchy-based 
regular graph-codes (see Theorem~\ref{thm:cauchy_gcode}). 
\vskip 5pt 
\begin{theorem} 
\label{thm:cauchy_pcode}
In the construction of an MBR product-matrix code~\cite{RashmiShahKumar2011} for a DSS $\D(n,k,d)$, if the encoding matrix $\Psi$ is a Cauchy matrix then 
the code is $(k-\ell)$-block secure against an adversary of strength $\ell$ $(\ell < k)$. 
\end{theorem}   
\begin{proof}[Sketch] 
To facilitate our discussion, we label the data units by the index set 
\begin{equation}
\label{eq:I}
\I = \{(i,j) \mid 1 \leq i \leq k \text{ and } i \leq j \leq d\}. 
\end{equation}
We assume that the $\M$ elements of $\I$ are always listed in the lexicographic order: $(1,1)$, $(1,2)$, $\ldots$, $(1,d)$, $(2,2)$, $(2,3)$, $\ldots$, $(2,d)$, $\ldots$,
$(k,k)$, $(k,k+1)$, $\ldots$, $(k,d)$. For $\xi = (i,j) \in \I$, we often write 
$f_\xi$ or $f_{(i,j)}$ interchangeably. Also, for $(i,j) \in \I$, both the $(i,j)$-entry
and the $(j,i)$-entry of $\bM$ are $f_{(i,j)}$. Again,
\vspace{-3pt}
\[
\bff = (f_{(1,1)},\ldots, f_{(1,d)}, f_{(2,2)}, \ldots, f_{(2,d)},\ldots, f_{(k,k)}, \ldots, 
f_{(k,d)})
\] 
is the vector in $\fq^\M$ that represents the file stored in the system.  
Suppose the encoding matrix $\Psi$ is a Cauchy matrix.
We assume that the adversary can access some $\ell$ storage nodes 
and $\bE$ is the submatrix of $\Psi$ consisting of the corresponding 
$\ell$ rows of $\Psi$. Hence the adversary obtains 
\vspace{-3pt}
\begin{equation} 
\label{eq:adv1}
\bH = \bE \bM.
\end{equation}  
This is regarded as a collection of $d$ linear systems with unknowns $f_\xi$'s, $\xi \in \I$. 
In order to apply Lemma~\ref{lem:basic}, we have to transform these systems
into a single linear system with unknown $\bff$. Let $\ee = (\bb_{j,\xi})_{j \in [d], \xi \in \I}$
be a $(d\ell) \times \M$ block matrix where
\begin{equation} 
\label{eq:ee}
\bb_{j,\xi} = 
\begin{cases}
\bE[i], &\text{if } \xi = (i,j) \text{ or } \xi = (j,i),\\
\bO, &\text{otherwise},
\end{cases}
\end{equation} 
and $\bO$ denotes the all-zero vector in $\fq^{\ell}$. 
Recall that $\bE[i]$ denotes the column $i$ of $\bE$. 
Let $\he = (\bH[1]^{\text{t}}, \bH[2]^{\text{t}}, \ldots, \bH[d]^{\text{t}})^{\text{t}}$. 
The proofs of the following two lemmas can be found in the Appendix. 
\begin{lemma} 
\label{lem:1}
The systems (\ref{eq:adv1}) are equivalent to the following system
\begin{equation} 
\label{eq:adv2}
\he = \ee \bff^{\text{t}}. 
\end{equation}
\end{lemma} 
\begin{lemma} 
\label{lem:2}
The linear code generated by the rows of the matrix $\ee$ has minimum distance $k-\ell + 1$. 
\end{lemma} 
Combining these two lemmas and Lemma~\ref{lem:basic} (applied to $\ee$ instead of $\bE$), 
we deduce that the code is $(k-\ell)$-block secure against adversaries of strength $\ell$ $(\ell < k)$. 
\end{proof} 
\vskip 10pt 

For instance, according to Theorem~\ref{thm:cauchy_pcode}, for $n = 7$, $k = 5$, $d = 6$, 
the degradation of the block security level of a Cauchy-based product-matrix code is illustrated in Fig.~\ref{fig:block_security_level_pm}. 
Note that the file size is $\M = 5 \times 6 - \binom{5}{2} = 20$.

\begin{figure}[htb]
\centering
\includegraphics{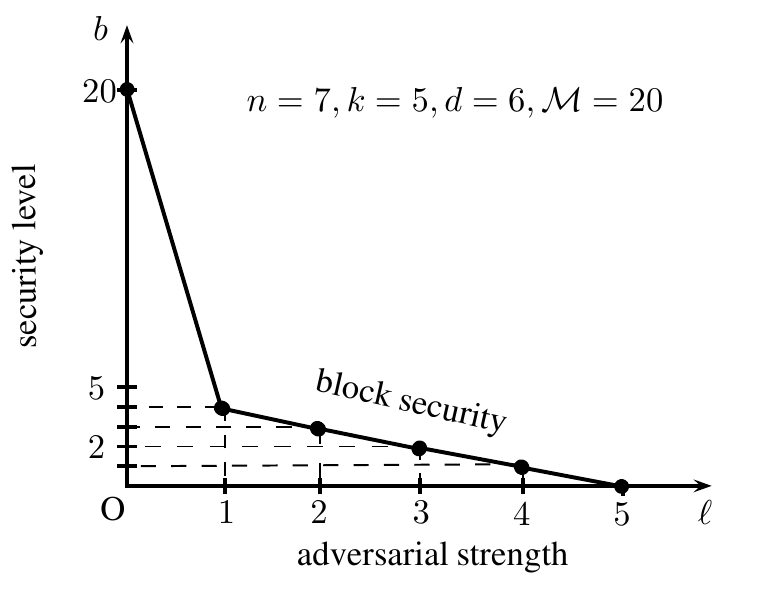}
\caption{Degradation of the security level for $\D(7, 5, 6)$ with a Cauchy-based product-matrix code as the adversarial strength increases.}
\label{fig:block_security_level_pm}
\end{figure}

Similar to Corollary~\ref{cr:gcode}, we can establish that if Cauchy matrices
are used in the construction of perfectly secure product-matrix codes~\cite{ShahRashmiKumar2011}, 
the codes remain $(k-\ell)$-block secure even when 
the adversarial strength $\ell$ surpasses the specified threshold for
perfect security $(\ell < k)$.  
For example, for a DSS $\D(7, 5, 6)$ as in Fig.~\ref{fig:block_security_level_pm}, 
suppose that the specified security threshold
is $\lambda = 2$, then the degradation of the security is depicted in Fig.~\ref{fig:perfect_block_pm}. 
Note that now the file size has to be decreased to $\M^{\text{s}} = 9$, according to (\ref{eq:pawar_bound}).  

\begin{figure}[htb]
\centering
\includegraphics{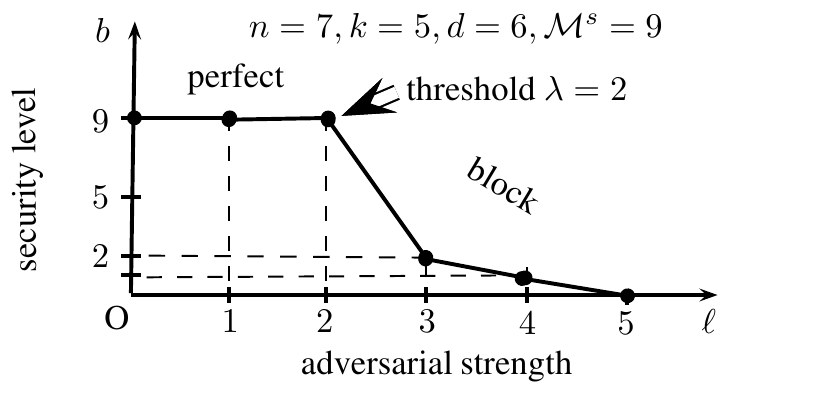}
\caption{Degradation of the security level for $\D(7, 5, 6)$ with a perfectly secure Cauchy-based product-matrix code 
as the adversarial strength increases.}
\label{fig:perfect_block_pm}
\end{figure}

\subsection{Vandermonde-Based Product-Matrix Codes Are Not Block Secure}
\label{subsec:vpcode}
 
The Vandermonde-based product-matrix codes are not $1$-block (weakly) secure
against an adversary of strength $k-1$. Indeed, consider the example 
from the original paper on product-matrix codes in Fig.~1, Section IV~\cite{RashmiShahKumar2011}. 
In that example, $k = 3$. By observing two nodes, Node $1$ and Node $6$, the adversary obtains two linear combinations
$f_7 + f_8 + f_9$ and $f_7+6f_8+f_9$. From these two, the adversary 
can deduce $f_8$. Hence, the code used in this example, which is 
based on a Vandermonde matrix, is not weakly secure against an adversary
of strength two. The Cauchy-based code, however, is weakly secure against an adversary
of strength two. 

\section{Conclusion}
\label{sec:conclusion}

We have evaluated the security levels of the two well-known MBR 
regenerating codes with exact repair, namely, regular-graph codes 
and product-matrix codes, employing a more general concept of security - block security. Block security provides weaker levels of security compared with perfect security, however, does not require a trade-off with the storage capacity. We established that Cauchy matrices play a vital role in 
guaranteeing block security for these MBR codes. Examining the security levels of known MSR regenerating codes is an interesting direction for future work.  

\bibliographystyle{IEEEtran}
\bibliography{Block-Secure-DSS}

\section*{Appendix}

\subsection{Proof of Lemma~\ref{lem:basic}}

The \emph{only if} direction is proved earlier in the discussion below Lemma~\ref{lem:basic}. 
It remains to prove the \emph{if} direction. 

Let $\tilde{\bff}$ be the guessed value (by the adversary) for the real file vector $\bff$. 
By eavesdropping, the adversary obtains some linear combinations of the coordinates
of $\tilde{\bff}$ and $\bff$, namely $\bu = \bE \tilde{\bff}^\text{t} = \bE \bff^\text{t}$.
We aim to show that for every nontrivial vector $\bv \in \fq^\M$ of weight less than $\dist(\cce)$, for the 
adversary, all possible values for the linear combination $\bv \tilde{\bff}^\text{t}$ are equally probable. 
As a consequence, the adversary gains no information about $\bv \bff^{\text{t}}$. 

Let $\bv \in \fq^\M$ be any nontrivial vector of weight less than $\dist(\cce)$ and $\tilde{s}$ an arbitrary element of $\fq$. 
It suffices to show that the system of linear equations
\begin{equation} 
\begin{cases}
\bE \tilde{\bff}^\text{t} = \bu&\\
\bv \tilde{\bff}^\text{t} = \tilde{s}& 
\end{cases}.
\label{eq:n1}
\end{equation} 
always has the same number of solutions $\tilde{\bff}$ for every choice of $\tilde{s} \in \fq$. 
It is a basic fact from linear algebra that the solution set for the system (\ref{eq:n1}) above, if nonempty, 
is an affine space, which is the sum of one solution of (\ref{eq:n1}) and the solution space of the corresponding 
homogeneous system. Therefore, if the system (\ref{eq:n1}) always has at least one solution for every
$\tilde{s}$, then it would have the same number of solutions for every $\tilde{s}$.  
Therefore, it remains to prove that this system always has a solution for every choice of $\tilde{s} \in \fq$. 

Indeed, let $s = \bv \bff^\text{t}$. Note that $s$ can be different from $\tilde{s}$.
We have
\begin{equation} 
\begin{cases}
\bE \bff^\text{t} = \bu&\\
\bv \bff^\text{t} = s& 
\end{cases}.
\label{eq:n2}
\end{equation} 
By subtracting the corresponding equations in the
two systems (\ref{eq:n1}) and (\ref{eq:n2}) and let $\bx = (\tilde{\bff}-\bff)$
be the new unknowns, we obtain the following system 
\begin{equation} 
\begin{cases}
\bE \bx^\text{t} = \bO&\\
\bv \bx^\text{t} = \tilde{s} - s& 
\end{cases}.
\label{eq:n3}
\end{equation}
It is clear that the system (\ref{eq:n1}) has a solution if and only if the system (\ref{eq:n3}) has 
a solution. Hence, it suffices to show that the system (\ref{eq:n3}) always has a solution for every
choice of $\tilde{s} \in \fq$. 

We claim that there exists some $\overline{\bx} \in \fq^\M$ satisfying $\bE  \overline{\bx}^\text{t} = \bO$
and $\bv \overline{\bx}^\text{t} \neq 0$. Then it is obvious that 
\[
\bx^* = \frac{\tilde{s} - s}{\bv \overline{\bx}^\text{t}}\overline{\bx}
\]
would be a solution of (\ref{eq:n3}), and hence the proof would follow. 
Indeed, if $\bv \bx^\text{t} = 0$ for every $\bx$ satisfying $\bE \bx^\text{t} = \bO$
then $\bv$ must belong to the orthogonal complement of the solution space of 
the system $\bE \bx^\text{t} = \bO$, which is precisely the row space of $\bE$. 
However,  the row space of $\bE$ contains no nontrivial vector of
weight less than $\dist(\cce)$. This would cause a contradiction since we assume 
from the beginning that $\bv$ is nontrivial and has weight less than $\dist(\cce)$. 
Thus, there must exist some $\overline{\bx} \in \fq^\M$ satisfying $\bE  \overline{\bx}^\text{t} = \bO$
and $\bv \overline{\bx}^\text{t} \neq 0$, as claimed above.

\subsection{Proof of Lemma~\ref{lem:1}}

We can rewrite the systems (\ref{eq:adv1}) as 
\[
\bH[j] = \bE (\bM[j]), \quad j \in [d]. 
\]
For each $j \in [d]$, the above system is equivalent to
\begin{equation}
\label{eq:basic}
\bH[j] = [\bb_{j, \xi_1} \mid \bb_{j,\xi_2} \mid \cdots \mid \bb_{j,\xi_\M}] \bff,
\end{equation}
where $\bb_{j,\xi}$ is given by (\ref{eq:ee}), $\{\xi_1,\xi_2,\ldots, \xi_\M\}$ forms the index set $\I$
defined in (\ref{eq:I}) in the lexicographic order, and $\bff = (f_{\xi_1}, f_{\xi_2}, \ldots, f_{\xi_\M})$ is the stored file. 
By vertically concatenating the systems (\ref{eq:basic}) for all $j \in [d]$, we obtain (\ref{eq:adv2}). 

As an example, let $k=3$ and $d=5$. Then $\M=12$ and the index set $\I$
is
\[
\begin{split}
\I &=\{\xi_1,\xi_2,\ldots, \xi_{12}\}\\
&=\big\{(1,1),(1,2),(1,3),(1,4),(1,5),(2,2),(2,3),(2,4),\\
&\qquad (2,5),(3,3),(3,4),(3,5)\big\}.
\end{split}
\]
The the message matrix is
\[
\bM = \begin{pmatrix}
f_{(1,1)} & f_{(1,2)}  & f_{(1,3)}  & f_{(1,4)}  & f_{(1,5)}\\
f_{(1,2)} & f_{(2,2)}  & f_{(2,3)}  & f_{(2,4)}  & f_{(2,5)}\\
f_{(1,3)} & f_{(2,3)}  & f_{(3,3)}  & f_{(3,4)}  & f_{(3,5)}\\
f_{(1,4)} & f_{(2,4)}  & f_{(3,4)}  & 0  & 0\\
f_{(1,5)} & f_{(2,5)}  & f_{(3,5)}  & 0  & 0\\
\end{pmatrix}.
\]
The file is 
\[
\begin{split}
\bff = (f_{(1,1)},f_{(1,2)},f_{(1,3)},f_{(1,4)},&f_{(1,5)},f_{(2,2)},f_{(2,3)},f_{(2,4)},\\ 
&f_{(2,5)},f_{(3,3)},f_{(3,4)},f_{(3,5)}).
\end{split}
\]
\setcounter{MaxMatrixCols}{15}
The matrix $\ee$ is (to save space we replace $\bE[i]$ by $\be_i$, columns are
indexed by $\xi_1,\xi_2,\ldots, \xi_{12}$)
{\footnotesize
\[
\ee = 
\bordermatrix{
& \xi_1&\xi_{2}&\xi_{3}&\xi_{4}&\xi_{5}&\xi_{6}&\xi_{7}&\xi_{8}&\xi_{9}&\xi_{10}&\xi_{11}&\xi_{12}\cr
&\be_1 & \be_2 & \be_3 & \be_4 & \be_5 & \bO & \bO & \bO & \bO & \bO & \bO & \bO\cr
&\bO & \be_1 & \bO & \bO & \bO & \be_2 & \be_3 & \be_4 & \be_5 & \bO & \bO & \bO\cr
&\bO & \bO & \be_1 & \bO & \bO & \bO & \be_2 & \bO & \bO & \be_3 & \be_4 & \be_5\cr
&\bO & \bO & \bO & \be_1 & \bO & \bO & \bO & \be_2 & \bO & \bO & \be_3 & \bO\cr
&\bO & \bO & \bO & \bO & \be_1 & \bO & \bO & \bO & \be_2 & \bO & \bO & \be_3
}.
\]
}

\subsection{Proof of Lemma~\ref{lem:2}}
Let $\ccee$ be the error-correcting code generated by the rows of $\ee$. 
We aim to prove that for every codeword of $\ccee$, each group of $k-\ell$ coordinates
can be represented as linear combinations of the other $\M - (k-\ell)$ coordinates, which do not belong to the group.
As a consequence, if a codeword $\cee \in \ccee$ has weight at most
$k - \ell$, i.e. it has some $\M - (k-\ell)$ zero coordinates, then 
its remaining $k-\ell$ coordinate must also be zero, and hence $\cee = \bO$. 
Therefore, every nonzero codeword in $\ccee$ has weight at least $k - \ell + 1$. 
Hence $\dist(\ccee) \geq k - \ell + 1$. This constitutes the most challenging
part in the proof of Lemma~\ref{lem:2}. The proof of the other direction, namely 
$\dist(\ccee) \leq k - \ell + 1$, is almost obvious, as we shall see later. 

We now introduce some necessary notations for the proof and establish 
the relationship among them.  
Let $\cce$ be the error-correcting code generated by the rows of $\bE$. 
Since $\bE$ is an $\ell\times d$ Cauchy matrix, $\cce$ is a $[d,\ell]$ MDS code (see Proof of Theorem~\ref{thm:cauchy_gcode}). 
Also, any set of $\ell$ columns of $\bE$ generate the whole column space. 
Therefore, for any $\ell$-subset $L$ of $[d]$ and any index $i \in [d] \setminus L$, we have 
\[
\bE[i] = \sum_{s \in L} a_{s,i}(L) \bE[s], 
\]
for some coefficients $a_{s,i}(L) \in \fq$. 
As a consequence, for any codeword $\ce = (c^{\bE}_1, \ldots, c^{\bE}_d) \in \cce$, we have
\begin{equation} 
\label{eq:ce}
c^{\bE}_i = \sum_{s \in L} a_{s,i}(L) c^{\bE}_s \quad (i \in [d] \setminus L).  
\end{equation} 
Note that as any set of $\ell$ columns of $\bE$ is linear independent, the coefficients $a_{s,i}(L)$ are all nonzero and uniquely determined by $\bE$ and $L$. From now on we only consider the case $L \subset [k]$. 
We often write $a_{i,j}$ instead of $a_{i,j}(L)$ to simplify the notation. 

For $j \in [d]$, let $\ccj$ be the row space of the matrix $\bB^{j} = [\bb_{j, \xi_1} \mid \bb_{j,\xi_2} \mid \cdots \mid \bb_{j,\xi_\M}]$.  
Then $\ccee = \cC^{1} + \cC^{2} + \cdots + \cC^{d}$ as a sum of spaces.
Note that according to the definition of $\bb_{j,\xi}$ in (\ref{eq:ee}),
\begin{itemize}
\item 
for $1 \leq j \leq k$: $\bb_{j,\xi} = \bE[i]$ if $\xi = (i,j) \in \I$, i.e. $1 \leq i \leq j$,
or $\xi = (j,i) \in \I$, i.e. $j+1 \leq i \leq d$,
\item 
for $k < j \leq d$: $\bb_{j,\xi} = \bE[i]$ if $\xi = (i,j) \in \I$, i.e. $1 \leq i \leq k$.  
\end{itemize}
Therefore, when $j \leq k$, the matrix $\bB^{j}$ has precisely $d$ nonzero columns, which are the same as
$d$ columns of $\bE$. 
On the other hand, when $j > k$, $\bB^{j}$ has precisely $k$ nonzero columns, which 
are the same as the first $k$ columns of $\bE$. 
Hence, for each codeword $\cj = (c^{j}_{\xi_1}, \ldots, c^{j}_{\xi_\M}) \in \ccj$ and each $\xi \in \I$, we have
\begin{equation} 
\label{eq:cj}
c^{j}_{\xi} = 
\begin{cases}
c^{\bE}_{i,j},& \text{ if } \xi = (i,j) \text{ or } \xi = (j,i),\\
0,& \text{ otherwise,}
\end{cases}
\end{equation}
where 
\begin{itemize}
\item
$(c^{\bE}_{1,j}, c^{\bE}_{2,j}, \ldots, c^{\bE}_{d,j})$ is a codeword of $\cce$, if $1 \leq j \leq k$,
\item 
$(c^{\bE}_{1,j}, c^{\bE}_{2,j}, \ldots, c^{\bE}_{k,j})$ is a codeword of $\cce$ restricted on the first $k$
coordinates, if $k < j \leq d$.
\end{itemize}
Since $\ccee = \cC^{1} + \cC^{2} + \cdots + \cC^{d}$, for each codeword $(c^{\ee}_{\xi_1},c^{\ee}_{\xi_2},\ldots,
c^{\ee}_{\xi_\M}) \in \ccee$ and each $\xi \in \I$, we have
\begin{equation} 
\label{eq:c}
c^{\ee}_{\xi} = 
\begin{cases}
c^{\bE}_{j,j},& \text{ if } \xi = (j,j),\\
c^{\bE}_{i,j} + c^{\bE}_{j,i},& \text{ if } \xi = (i,j) \text{ and } i < j.
\end{cases}
\end{equation}  
Note that $c^{\bE}_{i,j}$ is the $i$th coordinate of a codeword $\ce_j = (c^{\bE}_{1,j}, c^{\bE}_{2,j}, \ldots, c^{\bE}_{d,j})$
of $\cce$ and $c^{\bE}_{j,i}$ is the $j$th coordinate of another codeword $\ce_i = (c^{\bE}_{1,i}, c^{\bE}_{2,i}, \ldots, c^{\bE}_{d,i})$
of $\cce$. 

Due to (\ref{eq:ce}), for every $\ell$-subset $L$ of $[k]$ and 
for every $j \in [d]$ we have
\begin{equation} 
\label{eq:cej}
c^{\bE}_{i,j} = \sum_{s \in L} a_{s,i} c^{\bE}_{s,j},\quad \text{ for } i \in [d] 
\setminus L.  
\end{equation}

We now prove that $\dist(\ccee) \leq k - \ell + 1$. 
\begin{itemize}
\item If $d > k$ then 
$\dist(\ccj) = k - \ell + 1$, for every $j > k$, $j \leq d$.
Indeed, according to the definition, $\ccj$ is the row space of the matrix $\bB^{j}$. Moreover, according to the discussion above, for $j > k$, the matrix
$\bB^{j}$ has precisely $k$ nonzero columns, which are the same as the first
$k$ columns of $\bE$. These $k$ nonzero columns form an $\ell \times k$
Cauchy matrix, which in turn generates a $[k,\ell]$ MDS code of minimum distance $k - \ell + 1$ (see Proof of Theorem~\ref{thm:cauchy_gcode}). As other columns of $\bB^{j}$ are all-zero columns, 
we conclude that $\dist(\ccj) = k - \ell + 1$. As $\ccee = \cC^{1} + \cC^{2} + \cdots + \cC^{d}$ as a sum of spaces, we deduce that $\dist(\ccee) \leq k - \ell + 1$.
\item If $d = k$, then similar argument applies to $\ccj$ for any $j \in [k]$. 
\end{itemize} 
Thus, in both cases, $\dist(\ccee) \leq k - \ell + 1$.

We now proceed to the most important part of the proof of Lemma~\ref{lem:2}. 
Our goal is to show that $\dist(\ccee) \geq k - \ell + 1$.
Let $\cee = (c^{\ee}_{(1,1)}, c^{\ee}_{(1,2)},
\ldots, c^{\ee}_{(k,d)})$ be any codeword of $\ccee$. 
Let $U$ be any subset of $k - \ell$ elements of the index set 
$\I = \{(1,1),(1,2),\ldots, (k,d)\}$. 
We prove below that the coordinates of $\cee$ indexed by the elements of $U$
can be represented as linear combinations of the coordinates indexed by 
the elements of $\overline{U} = \I \setminus U$.
According to the discussion at the beginning of the proof of Lemma~\ref{lem:2},
this implies that $\dist(\ccee) \geq k - \ell + 1$. 
It suffices to show that for every index $(s,t) \in U$, 
the corresponding coordinate $c^{\ee}_{(s,t)}$ can be written as a linear
combination of $c^{\ee}_{(i,j)}$'s where $(i,j) \in \overline{U}$. 
We divide the proof into different cases, depending on whether $s = t$
or $s \neq t$. Three additional lemmas are introduced below to tackle those
cases separately. 
We henceforth drop $\ee$ from the notation $\cee$ to simplify the presentation.  

\begin{lemma}
\label{lem:case1} 
Let $\bc = (c_{(1,1)}, c_{(1,2)},\ldots, c_{(k,d)})$ be an arbitrary codeword of $\ccee$.
Suppose that $(t,t) \in U$ $(t \in [k])$. 
Then there exists an $\ell$-subset $L$ of $[k]$
such that
\begin{equation} 
\label{eq:ckk}
c_{(t,t)} = \sum_{i \in L}a_{i,t}(L)c^*_{(i,t)} 
- \sum_{\genfrac{}{}{0pt}{}{i \leq j}{i,j \in L}} a_{i,t}(L)a_{j,t}(L)c_{(i,j)}, 
\end{equation} 
where \vspace{-10pt}
\[
c^*_{(i,t)} =
\begin{cases}
c_{(i,t)}, &\text{ if } i \leq t,\\
c_{(t,i)}, &\text{ otherwise}.
\end{cases}
\]
Moreover, none of the indices $(i,t)$ (or $(t,i)$) and $(i,j)$ are in $U$
for every $i \in L$ and $j \in L$.  
\end{lemma}

According to Lemma~\ref{lem:case1}, $c_{(t,t)}$ can be written
as a linear combination of the coordinates of $\cee$ indexed by the elements
in $\overline{U}$, which is precisely what we want to show. 
Hence the case $(s=t,t) \in U$ is settled.
\begin{proof}[Proof of Lemma~\ref{lem:case1}]
We first construct an appropriate subset $L$ such that 
none of the sub-indices of the terms in the right-hand side of (\ref{eq:ckk}) belong to $U$.
Once such a subset is chosen, we proceed to prove that (\ref{eq:ckk}) indeed holds.
\begin{table}[h]
	\centering
		\begin{tabular}{|l|c|c|}
					\hline 
					& Elements & Count\\
					\hline 
					Starts with $t$ &  $(t, u_1), (t, u_2), \ldots, (t, u_a = t)$  &  $a$\\
					\hline
					Ends with $t$ & $(s_1, t), (s_2,t), \ldots, (s_b = t, t)$ & $b$\\
					\hline
					Not starts nor ends with $t$ & $(p_1, q_1), (p_2, q_2), \ldots, (p_c, q_c)$ & $c$ \\
					\hline
		\end{tabular}
		\caption{All elements of $U$ - Case 1.}
		\label{tab:case1}
\end{table}
\vspace{-5pt}

We list all elements of $U$ in Table~\ref{tab:case1}. 
Note that as the element $(t,t)$ is counted twice, we have $|U| = k - \ell=a+b+c-1$. 
It is straightforward that if $L$ does not contain any element in the 'bad' set
\[
L^* = \{u_1,\ldots,u_a = t\} \cup \{s_1, \ldots, s_{b-1}\} \cup \{q_1,\ldots, q_c\},
\] 
then none of the sub-indices of the terms in the right-hand side of (\ref{eq:ckk}) belong to $U$, as desired. As $|L^*| \leq a + (b-1) + c = k - \ell$, the set $[k] \setminus L^*$
has cardinality at least $\ell$. Therefore we can choose a legitimate $L$ by taking an arbitrary subset of $\ell$ elements of $[k] \setminus L^*$.     

We now show that (\ref{eq:ckk}) holds for $L$ chosen as above. 
By (\ref{eq:c}) and (\ref{eq:cej}) we have
\begin{equation} 
\label{eq:1}
\begin{split}  
\sum_{i \in L}a_{i,t}c^*_{(i,t)} 
&\underset{i \neq t}{=} \sum_{i \in L}a_{i,t}(c^{\bE}_{i,t}+c^{\bE}_{t,i})\\
&= \sum_{i \in L}a_{i,t}c^{\bE}_{i,t}
+\sum_{i \in L}a_{i,t}c^{\bE}_{t,i}\\
&= c^{\bE}_{t,t}
+\sum_{i \in L}a_{i,t}(\sum_{j \in L}a_{j,t}c^{\bE}_{j,i})\\
&= c_{(t,t)}+\sum_{i \in L}\sum_{j \in L}a_{i,t}a_{j,t}c^{\bE}_{j,i}\\
&= c_{(t,t)}+\sum_{i \in L}\sum_{j \in L}a_{i,t}a_{j,t}c^{\bE}_{i,j},
\end{split} 
\end{equation}
where in the last transition, the indices $i$ and $j$ are swapped.  
Also by (\ref{eq:c}) we have
\begin{equation} 
\label{eq:2}
\begin{split} 
&\quad \sum_{\genfrac{}{}{0pt}{}{i \leq j}{i,j \in L}} a_{i,t}a_{j,t}c_{(i,j)}\\
&= \sum_{\genfrac{}{}{0pt}{}{i = j}{i,j \in L}} a_{i,t}a_{j,t}c_{(i,j)}
+ \sum_{\genfrac{}{}{0pt}{}{i < j}{i,j \in L}} a_{i,t}a_{j,t} (c^{\bE}_{i,j} + c^{\bE}_{j,i})\\
&= \sum_{\genfrac{}{}{0pt}{}{i = j}{i,j \in L}} a_{i,t}a_{j,t}c^{\bE}_{i,j}
+\sum_{\genfrac{}{}{0pt}{}{i < j}{i,j \in L}} a_{i,t}a_{j,t} c^{\bE}_{i,j}
+\sum_{\genfrac{}{}{0pt}{}{i < j}{i,j \in L}} a_{i,t}a_{j,t} c^{\bE}_{j,i}\\
&= \sum_{\genfrac{}{}{0pt}{}{i = j}{i,j \in L}} a_{i,t}a_{j,t}c^{\bE}_{i,j}
+\sum_{\genfrac{}{}{0pt}{}{i < j}{i,j \in L}} a_{i,t}a_{j,t} c^{\bE}_{i,j}
+\sum_{\genfrac{}{}{0pt}{}{j < i}{i,j \in L}} a_{i,t}a_{j,t} c^{\bE}_{i,j}\\
&= \sum_{i \in L}\sum_{j \in L}a_{i,t}a_{j,t}c^{\bE}_{i,j},
\end{split} 
\end{equation} 
where the second to last transition is done by swapping the indices $i$
and $j$ in the third sum.  
Combining (\ref{eq:1}) and (\ref{eq:2}) we finish the proof of the Lemma~\ref{lem:case1} . 
\end{proof}
 \vskip 5pt 

We now consider $(s,t) \in U$ where $s \neq t$. We examine another two cases, 
depending on whether $t \leq k$ or $t > k$.   

\vskip 5pt 
\begin{lemma} 
\label{lem:case2}
Let $\bc = (c_{(1,1)}, c_{(1,2)},\ldots, c_{(k,d)})$ be an arbitrary codeword of $\ccee$.
Suppose that $(s,t) \in U$ and $t \in [k]$. 
Then there exists an $\ell$-subset $L$ of $[k]$ such that $s \in L$, $t \notin L$,
and
\begin{equation} 
\label{eq:claim2}
c_{(t,t)} = \sum_{i \in L}a_{i,t}(L)c^*_{(i,t)} 
- \sum_{\genfrac{}{}{0pt}{}{i \leq j}{i,j \in L}} a_{i,t}(L)a_{j,t}(L)c_{(i,j)}, 
\end{equation} 
where 
\[
c^*_{(i,t)} =
\begin{cases}
c_{(i,t)}, &\text{ if } i \leq t,\\
c_{(t,i)}, &\text{ otherwise}.
\end{cases}
\]
Moreover, for every $i \in L$ and $j \in L$, none of the indices $(i,t)$ (or $(t,i)$), except $(s,t)$, and $(i,j)$, except $(s,s)$, are in $U$.  
\end{lemma} 

As a consequence of Lemma~\ref{lem:case2}, $c_{(t,t)}$ can be written
as a linear combination of $c_{(s,t)}$, $c_{(s,s)}$, and the coordinates of $\cee$ indexed by the elements in $\overline{U}$. 
As noted before, the coefficient of $c_{(s,t)}$, namely $a_{s,t}$, is nonzero. 
Therefore, $c_{(s,t)}$ can also be written
as a linear combination of $c_{(t,t)}$, $c_{(s,s)}$, and the coordinates of $\cee$ indexed by the elements in $\overline{U}$.
Moreover, according to Lemma~\ref{lem:case1}, if $(t,t) \in U$ then $c_{(t,t)}$ can be written
as a linear combination of the coordinates of $\cee$ indexed by the elements in $\overline{U}$. Similar assertion holds for $c_{(s,s)}$ if $(s,s) \in U$. 
Thus, we conclude that $c_{(s,t)}$ can be written
as a linear combination of just the coordinates of $\cee$ indexed by the elements in $\overline{U}$.
Hence the case $(s,t) \in U$, $s \neq t \leq k$ is settled.
\begin{proof}[Proof of Lemma~\ref{lem:case2}]
We first construct an appropriate subset $L$ such that $s \in L$, $t \notin L$, and
none of the sub-indices of the terms in the right-hand side of (\ref{eq:claim2}) belong to $U$, except for $(s,t)$ and $(s,s)$.
Once such a subset is chosen, the proof that (\ref{eq:claim2}) indeed holds
is the same as that of Lemma~\ref{lem:case1}.
We divide the proof into four sub-cases, depending on whether $(s,s) \in U$
and/or $(t,t) \in U$.  

\nin\textbf{Sub-case 2-1:} $(s,s) \notin U$, $(t,t) \notin U$.  
We list all elements of $U$ in Table~\ref{tab:subcase1}. 
\begin{table}[H]
	\centering
		\begin{tabular}{|l|c|c|}
					\hline 
					& Elements & Count\\
					\hline 
					Starts with $s$ & $(s,t_1), (s,t_2),\ldots,(s,t_a = t)$ & a\\
					\hline
					Ends with $s$ & $(r_1,s), (r_2,s), \ldots,(r_b,s)$ & b\\
					\hline
					Starts with $t$ & $(t, u_1), (t, u_2), \ldots, (t, u_c)$ &  $c$\\
					\hline
					Ends with $t$ & $(s_1=s, t), (s_2,t), \ldots, (s_d, t)$ & $d$\\
					\hline
					Not starts/ends with $s$,$t$ & $(p_1, q_1), (p_2, q_2), \ldots, (p_e, q_e)$ & $e$ \\
					\hline
		\end{tabular}
		\caption{All elements of $U$ - Sub-case 2-1.}
		\label{tab:subcase1}
\end{table}
Note that as the element $(s,t)$ is counted twice, we have $|U| = k - \ell=a+b+c+d+e-1$. The 'bad' set to be excluded is 
\[
\begin{split}
L^* &= \{t_1,\ldots,t_a = t\} \cup \{r_1, \ldots, r_b\} \cup \{u_1,\ldots, u_c\}\\
&\quad\ \cup \{s_2,\ldots,r_d\} \cup \{q_1,\ldots,q_e\}.
\end{split}
\] 
As $|L^*| \leq a + b + c + (d-1) + e = k - \ell$, the set $[k] \setminus L^*$
has cardinality at least $\ell$. 
Moreover, $s \in [k] \setminus L^*$. 
Therefore we can choose a legitimate $L$ by taking an arbitrary 
subset of $\ell$ elements of $[k] \setminus L^*$ that contains $s$.  \\

\nin\textbf{Sub-case 2-2:} $(s,s) \in U$, $(t,t) \in U$.  
We list all elements of $U$ in Table~\ref{tab:subcase2}. 
\begin{table}[h]
	\centering
		\begin{tabular}{|l|c|c|}
					\hline 
					& Elements & Count\\
					\hline 
					Starts with $s$ & $(s,t_1=s), (s,t_2),\ldots,(s,t_a = t)$ & a\\
					\hline
					Ends with $s$ & $(r_1=s,s), (r_2,s), \ldots,(r_b,s)$ & b\\
					\hline
					Starts with $t$ & $(t, u_1), (t, u_2), \ldots, (t, u_c = t)$ &  $c$\\
					\hline
					Ends with $t$ & $(s_1=s, t), (s_2,t), \ldots, (s_d = t, t)$ & $d$\\
					\hline
					Not starts/ends with $s$,$t$ & $(p_1, q_1), (p_2, q_2), \ldots, (p_e, q_e)$ & $e$ \\
					\hline
		\end{tabular}
		\caption{All elements of $U$ - Sub-case 2-2.}
		\label{tab:subcase2}
\end{table}
Note that as the elements $(s,t)$, $(s,s)$, and $(t,t)$ are counted twice, we have $|U| = k - \ell=a+b+c+d+e-3$. The 'bad' set to be excluded is 
\[
\begin{split}
L^* &= \{t_2,\ldots,t_a = t\} \cup \{r_2, \ldots, r_b\} \cup \{u_1,\ldots, u_{c-1}\}\\
&\quad\ \cup \{s_2,\ldots,s_{d-1}\} \cup \{q_1,\ldots,q_e\}.
\end{split}
\] 
As $|L^*| \leq (a-1) + (b-1) + (c - 1) + (d-2) + e < k - \ell$, the set $[k] \setminus L^*$
has cardinality greater than $\ell$. 
Moreover, $s \in [k] \setminus L^*$. 
Therefore we can choose a legitimate $L$ by taking an arbitrary subset of $\ell$ elements of $[k] \setminus L^*$ that contains $s$.\\    

\nin\textbf{Sub-case 2-3:} $(s,s) \in U$, $(t,t) \notin U$.  
We list all elements of $U$ in Table~\ref{tab:subcase3}. 
\begin{table}[h]
	\centering
		\begin{tabular}{|l|c|c|}
					\hline 
					& Elements & Count\\
					\hline 
					Starts with $s$ & $(s,t_1=s), (s,t_2),\ldots,(s,t_a = t)$ & a\\
					\hline
					Ends with $s$ & $(r_1=s,s), (r_2,s), \ldots,(r_b,s)$ & b\\
					\hline
					Starts with $t$ & $(t, u_1), (t, u_2), \ldots, (t, u_c)$ &  $c$\\
					\hline
					Ends with $t$ & $(s_1=s, t), (s_2,t), \ldots, (s_d, t)$ & $d$\\
					\hline
					Not starts/ends with $s$,$t$ & $(p_1, q_1), (p_2, q_2), \ldots, (p_e, q_e)$ & $e$ \\
					\hline
		\end{tabular}
		\caption{All elements of $U$ - Sub-case 2-3.}
		\label{tab:subcase3}
\end{table}
Note that as the elements $(s,t)$ and $(s,s)$ are counted twice, we have $|U| = k - \ell=a+b+c+d+e-2$. The 'bad' set to be excluded is 
\[
\begin{split}
L^* &= \{t_2,\ldots,t_a = t\} \cup \{r_2, \ldots, r_b\} \cup \{u_1,\ldots, u_c\}\\
&\quad\ \cup \{s_2,\ldots,s_d\} \cup \{q_1,\ldots,q_e\}.
\end{split}
\] 
As $|L^*| \leq (a-1) + (b-1) + c + (d-1) + e < k - \ell$, the set $[k] \setminus L^*$
has cardinality greater than $\ell$. 
Moreover, $s \in [k] \setminus L^*$. 
Therefore we can choose a legitimate $L$ by taking an arbitrary subset of $\ell$ elements of $[k] \setminus L^*$ that contains $s$. \\    

\nin\textbf{Sub-case 2-4:} $(s,s) \notin U$, $(t,t) \in U$.  
We list all elements of $U$ in Table~\ref{tab:subcase4}. 
\begin{table}[h]
	\centering
		\begin{tabular}{|l|c|c|}
					\hline 
					& Elements & Count\\
					\hline 
					Starts with $s$ & $(s,t_1), (s,t_2),\ldots,(s,t_a = t)$ & a\\
					\hline
					Ends with $s$ & $(r_1,s), (r_2,s), \ldots,(r_b,s)$ & b\\
					\hline
					Starts with $t$ & $(t, u_1), (t, u_2), \ldots, (t, u_c=t)$ &  $c$\\
					\hline
					Ends with $t$ & $(s_1=s, t), (s_2,t), \ldots, (s_d=t, t)$ & $d$\\
					\hline
					Not starts/ends with $s$,$t$ & $(p_1, q_1), (p_2, q_2), \ldots, (p_e, q_e)$ & $e$ \\
					\hline
		\end{tabular}
		\caption{All elements of $U$ - Sub-case 2-4.}
		\label{tab:subcase4}
\end{table}
Note that as the elements $(s,t)$ and $(t,t)$ are counted twice, we have $|U| = k - \ell=a+b+c+d+e-2$. The 'bad' set to be excluded is 
\[
\begin{split}
L^* &= \{t_1,\ldots,t_a = t\} \cup \{r_1, \ldots, r_b\} \cup \{u_1,\ldots, u_{c-1}\}\\
&\quad\ \cup \{s_2,\ldots,s_{d-1}\} \cup \{q_1,\ldots,q_e\}.
\end{split}
\] 
As $|L^*| \leq a + b + (c-1) + (d-2) + e < k - \ell$, the set $[k] \setminus L^*$
has cardinality greater than $\ell$. 
Moreover, $s \in [k] \setminus L^*$. 
Therefore we can choose a legitimate $L$ by taking an arbitrary subset of $\ell$ elements of $[k] \setminus L^*$ that contains $s$.     

Thus, we complete the proof of Lemma~\ref{lem:case2}. 
\end{proof}

\begin{lemma} 
\label{lem:case3}
Let $\bc = (c_{(1,1)}, c_{(1,2)},\ldots, c_{(k,d)})$ be an arbitrary codeword of $\ccee$.
Suppose that $(s,t) \in U$, $s \neq t$, and $t > k$. 
Then there exists an $\ell$-subset $L$ of $[k]$ such that 
\begin{equation} 
\label{eq:ckm}
\begin{split} 
c_{(s,t)} &= \sum_{i \in L}a_{i,t}(L)c^*_{(i,s)}
+\sum_{i \in L} a_{i,s}(L) c_{(i,t)}\\
&\quad - \sum_{\genfrac{}{}{0pt}{}{i \leq j}{i, j \in L}} 
\big(a_{i,t}(L)a_{j,s}(L)+a_{i,s}(L)a_{j,t}(L)\big)c_{(i,j)} 
\end{split} 
\end{equation}
where 
\[
c^*_{(i,s)} =
\begin{cases}
c_{(i,s)}, &\text{ if } i \leq s,\\
c_{(s,i)}, &\text{ otherwise}.
\end{cases}
\]
Moreover, for every $i \in L$ and $j \in L$, none of the indices $(i,t)$, $(i,s)$ (or $(s,i)$), and $(i,j)$, are in $U$.  
\end{lemma} 

According to Lemma~\ref{lem:case3}, for $(s,t) \in U$, $s \neq t$, 
$t > k$, the coordinate $c_{(s,t)}$ can be written
as a linear combination of the coordinates of $\cee$ indexed by the elements
in $\overline{U}$, which is precisely what we want to show. 
Hence this very last case is settled.
As a consequence, the proof of Lemma~\ref{lem:2} also follows.

\begin{proof}[Proof of Lemma~\ref{lem:case3}]
We first construct an appropriate subset $L$ such that 
none of the sub-indices of the terms in the right-hand side of (\ref{eq:ckm}) belong to $U$.
Once such a subset is chosen, we proceed to prove that (\ref{eq:ckm}) indeed holds.
We consider two sub-cases, namely $(s,s) \in U$ and $(s,s) \notin U$. 
Note that as $t > k$, according to the definition of $\I$ in (\ref{eq:I}), $(t,t)$ is
not a valid index, and hence we do not have to consider the element $(t,t)$. \\ 

\nin\textbf{Sub-case 3-1:} $(s,s) \notin U$. 
We list all elements of $U$ in Table~\ref{tab:subcase3-1}. 
Since $t > k$, note that there is no valid index of the form $(t,u)$, according
to the definition of $\I$ in (\ref{eq:I}). 
\begin{table}[h]
	\centering
		\begin{tabular}{|l|c|c|}
					\hline 
					& Elements & Count\\
			    Starts with $s$ & $(s,t_1), (s,t_2),\ldots,(s,t_a = t)$ & a\\
					\hline
					Ends with $s$ & $(r_1,s), (r_2,s), \ldots,(r_b,s)$ & b\\
					\hline
					Ends with $t$ & $(s_1=s, t), (s_2,t), \ldots, (s_c, t)$ & $c$\\
					\hline
					Not starts/ends with $s$,$t$ & $(p_1, q_1), (p_2, q_2), \ldots, (p_d, q_d)$ & $d$ \\
					\hline
		\end{tabular}
		\caption{All elements of $U$ - Sub-case 3-1.}
		\label{tab:subcase3-1}
\end{table}
Note that as the element $(s,t)$ is counted twice, we have $|U| = k - \ell=a+b+c+d-1$. 
It is straightforward that if $L$ does not contain any element in the 'bad' set
\[
\begin{split}
L^* &= \{t_1,\ldots,t_{a-1}\} \cup \{r_1, \ldots, r_b\} \cup \{s_1,\ldots, s_c\}\\
&\quad\ \cup \{q_1,\ldots,q_d\},
\end{split}
\]  
then none of the sub-indices of the terms in the right-hand side of (\ref{eq:ckm}) belong to $U$, as desired. 
Note that here we do not have to exclude $t$, because we are about to choose $L$
as a subset of $[k]$ whereas we already assume that $t \notin [k]$ for this case.  
As $|L^*| \leq (a-1) + b + c + d = k - \ell$, the set $[k] \setminus L^*$
has cardinality at least $\ell$. Therefore we can choose a legitimate $L$ by taking an arbitrary subset of $\ell$ elements of $[k] \setminus L^*$.   \\

\nin\textbf{Sub-case 3-2:} $(s,s) \in U$. 
We list all elements of $U$ in Table~\ref{tab:subcase3-2}. 
Since $t > k$, note that there is no valid index of the form $(t,u)$, according
to the definition of $\I$ in (\ref{eq:I}). 
\begin{table}[h]
	\centering
		\begin{tabular}{|l|c|c|}
					\hline 
					& Elements & Count\\
			    Starts with $s$ & $(s,t_1=s), (s,t_2),\ldots,(s,t_a = t)$ & a\\
					\hline
					Ends with $s$ & $(r_1=s,s), (r_2,s), \ldots,(r_b,s)$ & b\\
					\hline
					Ends with $t$ & $(s_1=s, t), (s_2,t), \ldots, (s_c, t)$ & $c$\\
					\hline
					Not starts/ends with $s$,$t$ & $(p_1, q_1), (p_2, q_2), \ldots, (p_d, q_d)$ & $d$ \\
					\hline
		\end{tabular}
		\caption{All elements of $U$ - Sub-case 3-2.}
		\label{tab:subcase3-2}
\end{table}
Since the elements $(s,t)$ and $(s,s)$ are counted twice, we have $|U| = k - \ell=a+b+c+d-2$. 
The 'bad' set to be excluded is
\[
\begin{split}
L^* &= \{t_1,\ldots,t_{a-1}\} \cup \{r_2, \ldots, r_b\} \cup \{s_2,\ldots, s_c\}\\
&\quad\ \cup \{q_1,\ldots,q_d\}.
\end{split}
\]  
Note that here we do not have to exclude $t$, because we are about to choose $L$
as a subset of $[k]$ whereas we already assume that $t \notin [k]$ for this case.  
As $|L^*| \leq (a-1) + (b-1) + (c-1) + d < k - \ell$, the set $[k] \setminus L^*$
has cardinality larger than $\ell$. Therefore we can choose a legitimate $L$ by taking an arbitrary subset of $\ell$ elements of $[k] \setminus L^*$.     

Thus, in both cases we can always find an $\ell$-subset $L$ of $[k]$ such that
none of the sub-indices of the terms in the right-hand side of (\ref{eq:ckm}) belong to $U$.
We now show that (\ref{eq:ckm}) holds for $L$ chosen as above. 
Instead of $a_{i,j}(L)$, we simply write $a_{i,j}$ as there is no possible confusion. 
Using (\ref{eq:c}) and (\ref{eq:cej}) we have
\begin{equation} 
\label{eq:3}
\begin{split} 
& \quad \sum_{i \in L}a_{i,t}c^*_{(i,s)}+\sum_{i \in L}a_{i,s}c_{(i,t)}\\
&= \sum_{i \in L}a_{i,t}(c^{\bE}_{i,s}+c^{\bE}_{s,i})
+\sum_{i \in L}a_{i,s}(c^{\bE}_{i,t}+c^{\bE}_{t,i})\\
&= \sum_{i \in L}a_{i,t}c^{\bE}_{i,s}+\sum_{i \in L}a_{i,t}c^{\bE}_{s,i}+\sum_{i \in L}a_{i,s}c^{\bE}_{i,t}
+\sum_{i \in L}a_{i,s}c^{\bE}_{t,i}\\
&= c^{\bE}_{t,s}+\sum_{i \in L}a_{i,t}c^{\bE}_{s,i}+c^{\bE}_{s,t}+\sum_{i \in L}a_{i,s}c^{\bE}_{t,i}\\
&= c_{(s,t)} + \sum_{i \in L}a_{i,t}c^{\bE}_{s,i}
+\sum_{i \in L}a_{i,s}c^{\bE}_{t,i}\\
&= c_{(s,t)} + \sum_{i \in L}\sum_{j \in L} a_{i,t}a_{j,s}c^{\bE}_{j,i}
+\sum_{i \in L} \sum_{j \in L} a_{i,s}a_{j,t}c^{\bE}_{j,i}\\
&= c_{(s,t)} + \sum_{i \in L}\sum_{j \in L} a_{j,t}a_{i,s}c^{\bE}_{i,j}
+\sum_{i \in L} \sum_{j \in L} a_{j,s}a_{i,t}c^{\bE}_{i,j},
\end{split} 
\end{equation} 
where in the last transition, we swap the indices $i$ and $j$. 
We now calculate the remaining sum in the right-hand side of (\ref{eq:ckm}). 
We first split it into two sums and transform each individually.
The first sum is 
\begin{equation} 
\begin{split} 
\label{eq:sum_1}
& \quad \sum_{\genfrac{}{}{0pt}{}{i \leq j}{i,j \in L}} a_{i,t}a_{j,s} c_{(i,j)}\\
&= \sum_{\genfrac{}{}{0pt}{}{i = j}{i,j \in L}} a_{i,t}a_{j,s} c_{(i,j)}
+\sum_{\genfrac{}{}{0pt}{}{i < j}{i,j \in L}} a_{i,t}a_{j,s} c_{(i,j)}\\
&= \sum_{\genfrac{}{}{0pt}{}{i = j}{i,j \in L}} a_{i,t}a_{j,s} c^{\bE}_{i,j}
+\sum_{\genfrac{}{}{0pt}{}{i < j}{i,j \in L}} a_{i,t}a_{j,s} (c^{\bE}_{i,j} + c^{\bE}_{j,i})\\
&= \sum_{\genfrac{}{}{0pt}{}{i = j}{i,j \in L}} a_{i,t}a_{j,s} c^{\bE}_{i,j}
+\sum_{\genfrac{}{}{0pt}{}{i < j}{i,j \in L}} a_{i,t}a_{j,s} c^{\bE}_{i,j}
+ \sum_{\genfrac{}{}{0pt}{}{j < i}{i,j \in L}} a_{j,t}a_{i,s}c^{\bE}_{i,j},
\end{split} 
\end{equation} 
where in the last transition, we swap the indices $i$ and $j$. 
Similarly, the second sum is
\begin{equation} 
\begin{split} 
\label{eq:sum_2}
& \quad \sum_{\genfrac{}{}{0pt}{}{i \leq j}{i,j \in L}} a_{i,s}a_{j,t} c_{(i,j)}\\
&= \sum_{\genfrac{}{}{0pt}{}{i = j}{i,j \in L}} a_{i,s}a_{j,t} c^{\bE}_{i,j}
+\sum_{\genfrac{}{}{0pt}{}{i < j}{i,j \in L}} a_{i,s}a_{j,t} c^{\bE}_{i,j}
+ \sum_{\genfrac{}{}{0pt}{}{j < i}{i,j \in L}} a_{j,s}a_{i,t}c^{\bE}_{i,j}. 
\end{split} 
\end{equation} 
Combining (\ref{eq:sum_1}) and (\ref{eq:sum_2}) and grouping suitable
sums together we deduce 
\begin{equation} 
\label{eq:4}
\begin{split} 
& \quad \sum_{\genfrac{}{}{0pt}{}{i \leq j}{i,j \in L}} (a_{i,t}a_{j,s}+a_{i,s}a_{j,t})c_{(i,j)}\\
&= \sum_{i \in L}\sum_{j \in L} a_{j,t}a_{i,s}c^{\bE}_{i,j}
+\sum_{i \in L} \sum_{j \in L} a_{j,s}a_{i,t}c^{\bE}_{i,j}. 
\end{split} 
\end{equation} 
From (\ref{eq:3}) and (\ref{eq:4}) we derive (\ref{eq:ckm}), hence complete
the proof of Lemma~\ref{lem:case3}. 
\end{proof}
\vskip 10pt 

By Lemma~\ref{lem:case1}, Lemma~\ref{lem:case2}, and Lemma~\ref{lem:case3}, we show that for every codeword $\cee = (c^{\ee}_{(1,1)}, c^{\ee}_{(1,2)},
\ldots, c^{\ee}_{(k,d)})$ of $\ccee$, every $k-\ell$ coordinates
of $\cee$ can be written as linear combinations of the remaining $\M - (k - \ell)$ coordinates. Thus the code $\ccee$ has minimum distance at least $k - \ell + 1$.
As we already establish that this code has minimum distance at most $k - \ell + 1$,
the proof of Lemma~\ref{lem:2} follows.
\end{document}